\documentclass[11pt]{article}

\usepackage{graphicx, amssymb, latexsym, amsfonts, amsmath, lscape, amscd,
amsthm, color, epsfig, mathrsfs, tikz, enumerate, parskip}

\newtheorem{theorem}{Theorem}[section]

\newtheorem{corollary}[theorem]{Corollary}

\newtheorem{lemma}[theorem]{Lemma}

\newcommand\DELETE[1]{}

\begin{document}

\title{{\bf On relative clique number of colored mixed graphs}}
\author{
{\sc Sandip Das}$^{a}$, {\sc Soumen Nandi}$^{b}$, {\sc Debdeep Roy}$^{a}$, {\sc Sagnik Sen}$^{c}$ \\
\mbox{}\\
{\small $(a)$ Indian Statistical Institute, Kolkata, India}\\
{\small $(b)$ Birla Institute of Technology \& Science Pilani, Hyderabad Campus, India}\\
{\small $(c)$ Ramakrishna Mission Vivekananda Educational and Research Institute, Belur Math, India}
}

\date{\today}

\maketitle

\begin{abstract}
An $(m, n)$-colored mixed graph is a graph having arcs of $m$ different colors and
edges of $n$ different colors. A graph homomorphism of an $(m, n$)-colored mixed
graph $G$ to an $(m, n)$-colored mixed graph $H$ is a vertex mapping such that if $uv$
is an arc (edge) of color $c$ in $G$, then $f(u)f(v)$ is also an arc (edge) of color $c$.
The ($m, n)$-colored mixed chromatic number of an $(m, n)$-colored mixed graph $G$,
introduced by Ne\v{s}et\v{r}il and Raspaud [J. Combin. Theory Ser. B 2000] is the order
(number of vertices) of the smallest homomorphic image of $G$. Later Bensmail,
Duffy and Sen [Graphs Combin. 2017] introduced another parameter related to the
$(m, n)$-colored mixed chromatic number, namely, the $(m, n)$-relative clique number
as the maximum cardinality of a vertex subset which, pairwise, must have distinct
images with respect to any colored homomorphism.

In this article, we study the $(m, n$)-relative clique number for the family of subcubic graphs, graphs with maximum degree $\Delta$, planar graphs and triangle-free planar
graphs and provide new improved bounds in each of the cases. In particular, for
subcubic graphs we provide exact value of the parameter.
\end{abstract}


\section{Introduction}
The concept of vertex coloring, chromatic number and graph homomorphism was generalised by Ne\v{s}et\v{r}il and Raspaud~\cite{impj} by defining $(m,n)$-colored mixed graphs and colored graph homomorphisms. This idea of homomorphism captures the definition of homomorphism for graphs, oriented graphs and edge-colored graphs.
	
	A mixed graph is a simple graph where a subset of the edges have been oriented to become arcs. An $(m,n)$-colored mixed graph $G$ with vertex set $V(G)$, arc set $A(G)$ and edge set $E(G)$ is a mixed graph where each arc is colored with one of the $m$ colors $\lbrace 1,2,3,\cdots,m\rbrace $ and each edge is colored with one of the $n$ colors $\lbrace 1,2,3,\cdots,n\rbrace $. When $m=0$ (resp. $n=0$), it is assumed that the $(0,n)$-colored mixed graph (resp. $(m,0)$-colored mixed graph) does not contain any arcs (resp. edges). So a $(1,0)$-colored mixed graph is an oriented graph, a $(0,1)$-colored mixed graph is an undirected graph and a $(0,k)$-colored mixed graph is a $k$-edge-colored graph~\cite{coxeter}.
	
	Every $(m,n)$-colored mixed graph $G$ has an underlying simple graph denoted by $U(G)$. If $uv\in E(U(G))$, then the adjacency type of $uv$ is an edge colored $i$ if $uv\in E(G)$ and $uv$ has color $i$ or an arc colored $j$ if $uv\in A(G)$ and $uv$ has color $j$. In this article, we consider only those $(m,n)$-colored mixed graphs whose underlying graphs are simple that is, they don't have any loops or multiple edges.
	
	In discussing an $(m,n)$-colored mixed graph, we make no distinction in notation between the edges and arcs of the graph. As each pair of adjacent vertices in $U(G)$ has exactly one adjacency type in $G$, there is no scope for confusion in the notation $uv$ being used to refer to either an arc from $u$ to $v$ or an edge between $u$ and $v$, as the case may be. We say that 
	$uv,wx\in A(G) \cup E(G)$ have the same adjacency type if one of the following holds:

	\begin{itemize}
	\item $uv,wx\in A(G)$ and both have color $i\in \lbrace1,2,3,\cdots,m\rbrace$,
	
	\item $vu,xw\in A(G)$ and both have color $i\in \lbrace1,2,3,\cdots,m\rbrace$,
	
	\item $uv,wx\in E(G)$ and both have color $j\in \lbrace1,2,3,\cdots,n\rbrace$.
	\end{itemize}

	Let $G$ and $H$ be $(m,n)$-colored mixed graphs. A $colored~homomorphism$ of $G$ to $H$ is a function 
	$f:V(G) \rightarrow V(H)$ such that the adjacency type of $uv$ in $G$ is the same as that of 
	$f(u)f(v)$ in $H$, for all $uv \in E(G) \cup A(G)$. In other words, a colored homomorphism is a vertex mapping that preserves colored edges and colored arcs \cite{impj}. We write $f:G\rightarrow H$ when there exists a homomorphism $f$ from $G$ to $H$ and $H$ is called the homomorphic image of $G$.
	
	The $(m,n)$-colored mixed chromatic number of $G$, denoted by $ \chi_{m,n}(G)$, is the least integer $k$ such that there exists a homomorphic image of $G$ of order $k$. For a simple graph $\Gamma $, we let $ \chi_{m,n}(\Gamma)$ denote the maximum $(m,n)$-colored mixed chromatic number over all $(m,n)$-colored mixed graphs $G$ such that $U(G)=\Gamma $. For a family $\mathcal{F}$ of 
	 graphs   $\chi_{m,n}(\mathcal{F})$ denote the maximum of $ \chi_{m,n}(G)$ taken over all $G \in \mathcal{F}$.
	
	It is observed that letting $m=0$ and $n=1$ in the definitions given above lead to the usual definitions of graph homomormorphism and chromatic number. Likewise, letting $m=1$ and $n=0$ gives the definition of oriented graph homomorphism and oriented chromatic number considered by many researchers over the last two decades. Also, putting $m=0$ and $n=k$ gives the definition of homomorphism used by many authors in the study of homomorphisms of $k$-edge-colored graphs \cite{coxeter, simon, ochem}.


	The main focus of this article is on structures analogous to cliques for general $(m,n)$-colored mixed graphs.
	
	An \textit{$(m,n)$-clique} $C$ is an $(m,n)$-colored mixed graph for which $ \chi_{m,n}(C)=|V(C)|$. The 
	\textit{$(m,n)$-absolute clique number}
	$\omega_{a(m,n)}(G)$ of an 
	$(m,n)$-colored mixed graph $G$ is the largest $k$ such that $G$ contains an $(m,n)$-clique of order $k$. For $m=0$ and $n=1$, the above definitions coincide with the definitions of clique and clique number and for $m=1$ and $n=0$. Also the definition above is the same as the definitions of oriented clique  and oriented absolute clique number introduced by Klostermeyer and MacGillivray~\cite{mac}.
	
	In the previous studies of 
	oriented cliques~\cite{sen, nandy}, a related parameter namely the \textit{oriented relative clique} number was an useful tool in the study of oriented chromatic number. A generalization of this parameter for 
	$(m,n)$-colored mixed graphs was introduced by Bensmail, Duffy and Sen~\cite{duffy}. 
	A vertex subset $R\subseteq V(G))$ is a \textit{relative $(m,n)$-clique} of an $(m,n)$-colored mixed graph $G$ if for every pair of distinct vertices $u,v\in R$ and every homomorphism $f:G\rightarrow H$, we have $f(u)\neq f(v)$. That is, no two distinct vertices of a relative clique can be identified under any homomorphism. The \textit{$(m,n)$-relative clique number} 
	$\omega_{r(m,n)}(G)$ of an $(m,n)$-colored mixed graph $G$ is the cardinality of a largest relative 
	$(m,n)$-clique of $G$. Bensmail, Duffy and Sen~\cite{duffy}
	showed that 
	$$\omega_{a(m,n)}(G)\leq \omega_{r(m,n)}(G)\leq \chi_{m,n}(G).$$ 
	For simple undirected graphs, the absolute and relative clique numbers coincide which is not the case when $(m,n)\neq (0,1)$~\cite{duffy}.

For a simple graph $\Gamma $ the $(m,n)$-absolute (relative)
clique number $\omega_{a(m,n)}(\Gamma)$ (respectively,
$\omega_{a(m,n)}(\Gamma)$) denote the maximum 
$(m,n)$-absolute (relative)
clique number over all $(m,n)$-colored mixed graphs $G$ such that $U(G)=\Gamma $. For a family $\mathcal{F}$ of graphs 
 $\omega_{a(m,n)}(\mathcal{F})$ 
 (respectively, $\omega_{a(m,n)}(\mathcal{F})$) denote the maximum of 
 $\omega_{a(m,n)}(G)$ (respectively, $\omega_{a(m,n)}(G)$) taken over all $G \in \mathcal{F}$.

Furthermore, Bensmail, Duffy and Sen~\cite{duffy} provided the following characterization of  $(m,n)$-cliques. Let $G$ be an $(m,n)$-colored mixed graph and let $uvw$ be a 2-path in $U(G)$. We say that $uvw$ is a $special~2-path$ if  one of the following statements hold:
	\begin{itemize}
	\item $uv$ and $vw$ are edges of different colors,
	
	\item $uv$ and $vw$ are arcs(possibly of same color),
	
	\item $uv$ and $wv$ are arcs of different colors,
	
	\item $vu$ and $vw$ are arcs of different colors,
	
	\item exactly one of $uv$ and $vw$ is an edge.
	\end{itemize}

In other words, $uvw$ is a \textit{special 2-path} if $uv$ and 
$vw$ do not have the same adjacency type.

In the following we rephrase a lemma and its corollary 
due to Bensmail, Duffy and Sen~\cite{duffy} which is instrumental to our work.

\begin{lemma}[Bensmail, Duffy and Sen 2017~\cite{duffy}]\label{lem structure}
Let $G$ be an $(m,n)$-colored mixed graph. A pair of vertices $u,v\in V(G)$ are part of a relative clique iff they are either adjacent or connected by a special 2-path.
\end{lemma}
\begin{proof}
(Necessity) Let $u,v$ be two vertices of an $(m,n)$-colored mixed graph $G$. In case they are not a part of any relative clique, then there exists an $(m,n)$-colored mixed graph $H$ and a homomorphism $f:G\rightarrow H$ such that $f(u)=f(v)$. If $u$ and $v$ are adjacent, then $f(u)f(v)$, the image of $uv$, is a loop in $H$ contradicting the fact that $U(H)$ is simple. Assume that $u,v$ are the ends of a special 2-path $uxv$. As $f(u)=f(v)$ and $U(H)$ is simple, so $ux$ and $vx$ must have similar adjacency type. This contradicts the fact that $uxv$ is a special 2-path. So. $u$ and $v$ can neither be adjacent, nor connected by a special 2-path.

(Sufficiency) Assume that $u$ and $v$ are neither adjacent nor connected by a special 2-path. Identifying vertices $u$ and $v$ and deleting duplicate edges/arcs of the same color, we arrive at an $(m,n)$-colored mixed graph $H$. Let $x$ be the vertex obtained by identifying vertices $u$ and $v$. Consider the vertex mapping $f:V(G)\rightarrow V(H)$ given by 
$$
f(z)=
\begin{cases}
x~ \text{if}~ z=u,v\\
z  ~  otherwise
\end{cases}
$$
The function $f$ is a homomorphism of $G$ to $H$.
\end{proof}

An immediate corollary due to 
Bensmail, Duffy and Sen~\cite{duffy} characterizes $(m,n)$-cliques. 

\begin{corollary}[Bensmail, Duffy and Sen 2017~\cite{duffy}]
An \textit{$(m,n)$-colored mixed graph} $G$ is an $(m,n)$-clique if and only if every pair of non-adjacent vertices of $G$ are joined by a special 2-path.
\end{corollary}

\begin{proof}
From the definitions, it follows that an $(m,n)$-colored mixed graph $G$ is an $(m,n)$-clique if and only if all its vertices are part of the same relative clique. The result is now immediate from Lemma 1.1.
\end{proof}	

Using the characterization presented in 
Lemma~\ref{lem structure} we study the structure of $(m,n)$-relative cliques and the lower/upper bounds of $(m,n)$-relative clique numbers of different graph families. 
	

In this article, we have generalized some results 
known for 
$(1,0)$-relative clique number (alternatively known as oriented relative clique number) due to 
Das, Mj and Sen~\cite{sir} for all $(m,n) \neq (0,1)$.

In Section~\ref{Chapter2_5} we provided a detailed literature review. 
In Section~\ref{Chapter3} we studied $(m,n)$-relative clique number for  graphs with maximum degree $\Delta$.
In Section~\ref{Chapter4} we studied $(m,n)$-relative clique number for subcubic graphs. 
In Section~\ref{Chapter5} we studied $(m,n)$-relative clique number of planar  graphs.  
Finally, in Section~\ref{chap con} we conclude the article. 
	
\section{Literature Survey}\label{Chapter2_5}
In this chapter, we list out previously known results related to the works done in this article. We have not maintained the standard chronological order while listing the results here. Rather we have clubbed the results which are related together for the sake of the flow of this 
review.

A  \textit{tree} is a connected graph with no cycles. A \textit{forest} is a graph whose each connected component is a tree.

\begin{theorem}[Ne\v{s}et\v{r}il and Raspaud 1998 \cite{impj}]
Let $\mathcal{F}_{m,n}$ be the class of 
$(m,n)$-colored mixed forests. Then

$$
\chi_{m,n}(\mathcal{F}_{m,n})=
\begin{cases}
2m+1, \text{ where}  n=0\\
2(m+\lfloor \frac{n}{2} \rfloor +1) \text{ for } n\neq 0.
\end{cases}
$$
\end{theorem}

A \textit{path} is a graph with a 
 sequence of vertices 
in which consecutive vertices are adjacent. 
Fabila-Monroy, Flores, Huemer and Montejano~\cite{flores}
calculated the exact $(m,n)$-colored mixed chromatic number of the family of all paths. 

\begin{theorem}[Fabila-Monroy, Flores, Huemer and Montejano 2008 \cite{flores}]
Let $\mathcal{L}$ be the class of paths. Then 
$$\displaystyle \chi_{m,n}(\mathcal{L})=2m+n+\epsilon,$$ where $ \epsilon =1 $ for $n$ odd or $n = 0$, and $ \epsilon =2 $ for $n > 0$ even.
\end{theorem}

A  \textit{$k$-acyclic coloring} of the vertices of
an undirected graph $G$ is an assignment of $k$ colors to $G$ such that each color induces an independent set and any two colors induces a forest.  The \textit{acyclic chromatic number} $\chi_a(G)$ of
a graph $G$ is the smallest $k$ such that $G$ has an acyclic $k$-coloring. 

\begin{theorem}[Ne\v{s}et\v{r}il and Raspaud 1998 \cite{impj}]\label{th acyclic}
If $G$ is an $(m,n)$-colored mixed graph for which the acyclic chromatic number of the underlying undirected graph is at most $k$, then 
$$\chi_{m,n}(G)\leq k (2m+n)^{(k-1)}.$$
\end{theorem}

Fabila-Monroy, Flores, Huemer and Montejano~\cite{flores} showed the tightness of the bound.

\begin{theorem}[Fabila-Monroy, Flores, Huemer and Montejano 2008 \cite{flores}]
For every $k\geq 3$ and every $m,n\geq 0$, 
$$\displaystyle \chi_{m,n}(\mathcal{A}_k)=k(2m+n)^{(k-1)},$$ where $\mathcal{A}_k$ is the family of graphs with acyclic chromatic number at most k.
\end{theorem}

Ne\v{s}et\v{r}il and Raspaud~\cite{impj}  showed that 
the $(m,n)$-colored mixed chromatic number 
of a graph is bounded by a function of its acyclic 
chromatic number. 
 Das, Nandi and Sen~\cite{nandi} showed that the reverse type of bound also exists, that is, 
 the acyclic 
chromatic number 
of a graph is bounded by a function of its 
$(m,n)$-colored mixed chromatic number.

\begin{theorem}[Das, Nandi and Sen 2016~\cite{nandi}]
Let $G$ be an $(m,n)$-colored mixed graph with $arb(G) = r$ and $\chi_{(m,n)}G=k$ where $p=(2m + n)\geq 2$. Then 
$$\displaystyle \chi_a(G)\leq k^2+k^{2+\lceil log_2~ log_p k\rceil}.$$
\end{theorem}

The \textit{arboricity} $arb(G)$ of a graph $G$ is the minimum $k$ such that the edges of $G$ can be decomposed into $k$ forests.

Das, Nandi and Sen~\cite{nandi} also showed that showed that the  
$(m,n)$-colored mixed chromatic number of a graph is 
bounded by a function of its acyclic chromatic number and arboricity which gives a better bound than the one given by 
the above result.

\begin{theorem}[Das, Nandi and Sen 2016~\cite{nandi}]
Let $G$ be an $(m,n)$-colored mixed graph with $arb(G) = r$ and $\chi_{(m,n)}G=k$ where $p=(2m + n)\geq 2$. Then $$\displaystyle \chi_a(G)\leq k^{\lceil \log_p r \rceil+1}.$$
\end{theorem}

 Das, Nandi and Sen~\cite{nandi} further showed that the  arboricity
of a graph is bounded by a function of its 
$(m,n)$-colored mixed chromatic number.

\begin{theorem}[Das, Nandi and Sen 2017~\cite{nandi}]
Let $G$ be an $(m,n)$-colored mixed graph with 
$\chi_{(m,n)}(G)=k$.
Then $$arb(G) \leq \lceil \log_p k+\frac{k}{2} \rceil$$
 where $p=(2m + n)\geq 2$.
\end{theorem}

A \textit{planar graph} is a graph that  can be drawn on the plane in such a way that its edges intersect only at the vertices. In other words, it can be drawn in such a way that no edges cross each other.

The celebrated result due to Borodin shows that a planar graph has acyclic chromatic number at most 5.

\begin{theorem}[Borodin 1979 \cite{boro}]\label{th borodin}
For the family $\mathcal{P}$ of planar graphs
$$\chi_a(\mathcal{P})\leq 5.$$
\end{theorem}

 Using Theorem~\ref{th acyclic} and~\ref{th borodin}, Ne\v{s}et\v{r}il and Raspaud  \cite{impj} established the following bound.

\begin{theorem}[Ne\v{s}et\v{r}il and Raspaud 1998 \cite{impj}]
For the family $\mathcal{P}$ of planar graphs
$$\chi_{m,n}(\mathcal{P})\leq 5(2m+n)^4.$$
\end{theorem}

A \textit{$k$-tree} is a simple graph obtained from the
complete graph $K_k$ by repeatedly inserting new vertices adjacent to all vertices of
an existing clique of order $k$. A \textit{partial $k$-tree} is a subgraph of some $k$-tree.

An \textit{outerplanar graph} is a planar graph that can be embedded in the plane without crossings in such
a way that all the vertices lie in the unbounded face of the embedding. It is known that an outerplanar graph is also a partial 2-tree.

\begin{theorem}[Fabila-Monroy, Flores, Huemer and Montejano 2008 \cite{flores}]
Denote the families of partial k-trees, outerplanar and planar graphs by $\mathcal{T}_{k},\mathcal{O}~and~\mathcal{P}$ respectively. Let $\epsilon =1$  for 
$n$ odd or 
$n = 0$, and $ \epsilon =2 $ for $n > 0$ even. Then:

\begin{itemize}
\item[$(i)$] $(2m+n)^{k}+\epsilon (2m+n)^{(k-1)}+(2m+n)^{(k-2)}+\cdots+1\leq \chi_{m,n}(\mathcal{T}^{k})$,

\item[$(ii)$] $(2m+n)^2+\epsilon(2m+n)+1\leq \chi_{m,n}(\mathcal{O})$

\item[$(iii)$] $(2m+n)^3+\epsilon(2m+n)^2+(2m+n)+1\leq \chi_{m,n}(\mathcal{P})$
\end{itemize}
\end{theorem}

Another interesting general result regarding $(m,n)$-colored mixed chromatic number is a bound with respect to 
the maximum degree of the graph.

\begin{theorem}[Das, Nandi and Sen 2016~\cite{nandi}]
For the family $G_{\Delta}$ of graphs with maximum degree $\Delta $, we have

 $$\displaystyle p^{\frac{\Delta}{2}}\leq \chi_{(m,n)}(\mathcal{G}_{\Delta})\leq 2(\Delta -1)^p p^{(\Delta -p+2)} +2$$ for all $p = 2m + n \geq2$ and for all $\Delta \geq 5$.
\end{theorem}


Finally we enter the territory of $(m,n)$-cliques. 
The first result due to 
Bensmail, Duffy and Sen~\cite{duffy} shows that 
$(m,n)$-cliques are not really rare objects.

\begin{theorem}[Bensmail, Duffy and Sen 2017 \cite{duffy}]
For $(m, n) \neq (0, 1)$ almost every $(m,n)$-colored mixed graph is an $(m,n)$-clique.
\end{theorem}

The next result is the first result concerning $(m,n)$-relative clique number.

\begin{theorem}[Bensmail, Duffy and Sen 2017 \cite{duffy}]
For the family $\mathcal{O}$ of outerplanar graphs, we have $$ \omega_{a(m,n)}(\mathcal{O}) = \omega_{r(m,n)}(\mathcal{O})=3(2m+n)+1$$ where $\omega_{a(m,n)}(G)$ and $\omega_{r(m,n)}(G)$ denote the absolute and relative clique numbers of an $(m,n)$-colored mixed graph $G$.
\end{theorem}

Using the fact that $\omega_{r(m,n)}(\mathcal{O})=3(2m+n)+1$ from the above result Bensmail, Duffy and Sen~\cite{duffy} proved the following interesting result regarding $(m,n)$-absolute clique number of planar graphs.

\begin{theorem}[Bensmail, Duffy and Sen 2017 \cite{duffy}]
For the family $\mathcal{P}$ of planar graphs, we have 
$$3(2m+n)^2+(2m+n)+1\leq \omega_{a(m,n)}(\mathcal{P}) \leq 9(2m+n)^2+2(2m+n)+2$$ for all $(m,n)\neq (0,1)$.
\end{theorem}

Moreover, Bensmail, Duffy and Sen~\cite{duffy} concluded their paper by implicitly suggesting that one may study 
the $(m,n)$-relative and absolute clique numbers using ideas from similar works done for 
$(m,n) = (1,0)$ or $(0,2)$. 
Their suggestions are a direct motivation of the works done in this article.

\section{Maximum Degree}\label{Chapter3}
We start by establishing an upper bound for $(m,n)$-relative clique number of the family $\mathcal{G}_\Delta$ of  graphs with maximum degree $\Delta$. The bound obtained is quadratic in 
$\Delta$.

\begin{theorem}\label{th rel delta}
 For the family $ \mathcal{G}_\Delta$ of  graphs with maximum degree $ \Delta\geq 0$, we have $$\displaystyle \omega_{r(m,n)}({\mathcal{G}_\Delta})\leq \lfloor\frac{p-1}{p}\Delta^2\rfloor+\Delta+1 $$
 where $p=2m+n > 1$.
 \end{theorem}

We shall need the following definitions and a supporting lemma to prove Theorem~\ref{th rel delta}.

An undirected simple graph $G$ is \textit{$k$-degenerate} if each subgraph of $G$ contains a vertex with degree at most $k$. 

In an $(m,n)$-colored mixed graph, a vertex $a$ is said to \emph{see} a vertex $b$ if they are either adjacent or connected by a special 2-path. If $a$ and $b$ are connected by a special 2-path with $w$ acting as the internal vertex, then it is said that $u~sees~v~through~w$ or equivalently $v~sees~u~through~w$.

If $G$ is an $(m,n)$-colored mixed graph, then $G^2$ is defined as the graph with set of vertices $V(G^2)=V(G)$ and set of edges $E(G^2)=\{uv|u \text{ sees } v \text{ in } G\}$. Note that $G^2$ is an undirected graph.

 \begin{lemma}\label{th rel polta}
 Let $G$ be an $(m,n)$-colored mixed graph with maximum degree $ \Delta$. Then its special 2-path graph $ G^2$=($V,E'$) is $ \lfloor{\frac{(2m+n-1)\Delta^2}{2m+n}}\rfloor+\Delta$-degenerate for all $(m,n) \neq (0,1)$.
 \end{lemma}

 \begin{proof}
 We use  discharging method to show that for any set $S$ $\subseteq V$, the induced graph $ {G}^2$[$S$] has minimum degree at most $$ \lfloor\frac{(2m+n-1)\Delta^2}{2m+n}\rfloor+\Delta .$$
 Assume that the initial charge $ \gamma(v)$ of the vertices to be
$$
\gamma(v)=
\begin{cases}
\frac{\Delta^2}{p},~ \text{if}~  v\in V(S)\\
~0~, ~ \text{if}~  v\notin V(S)
\end{cases}
$$
where $ p=2m+n$.
So the total charge of the graph is $ \frac{\Delta^2|S|}{p}$. Now we proceed to the discharging step where every vertex of $S$ gives the charge $ \frac{\Delta}{p}$ to each of its neighbouring vertices.

Let $ \gamma^*(v)$ be the new updated charge of the vertices of $G$. So a vertex $v$ with $k$ neighbours in $S$ has charge $ \gamma^*(v)\geq\frac{k\Delta}{p}$.
Now let us denote by $ \pi_s(v)$, the number of special 2-paths going through a vertex $v$ and linking two vertices in $S$. So for a vertex $v$ with $k$ neighbours in $S$, we have that 

\begin{align*}
\displaystyle \pi_s(v) \leq   & max\left\lbrace\sum_{1\leq j<l \leq 2m+n}i_j \times i_l \text{ subject to } \sum_{j=1}^{2m+n}i_j=k\right\rbrace\\
& \leq\lfloor{{2m+n}\choose{2}}\left(\frac{k}{2m+n}\right)^2\rfloor \\
& =\lfloor{{p}\choose{2}}\left(\frac{k}{p}\right)^2\rfloor=\lfloor\frac{(p-1)k^2}{2p}\rfloor
\end{align*}
 where $ i_{2j-1}$ and $ i_{2j}$ denote the number of incoming and outgoing arcs of color $ j,1\leq j\leq m$ and $ i_{j+2m}$ denotes the number of edges of color $ j,1\leq j\leq n$.
Since $k \leq \Delta$, we have that 
$$ \pi_s(v)\leq \lfloor\frac{(p-1)k^2}{2p}\rfloor\leq \frac{(p-1)k\Delta}{2p}\leq \frac{(p-1)\gamma^*(v)}{2}.$$

\bigskip

\noindent So the number of edges in $ {G}^2[S]$ is
$$\displaystyle E[ {G}^2[S]]\leq \frac{p-1}{2}\sum_{v\in V(G)}\gamma^*(v)+\frac{\Delta}{2}|S|=\frac{p-1}{2}\sum_{v\in V(G)}\gamma(v)+\frac{\Delta}{2}|S|$$ as the discharging does not change the total charge of the graph.

 \bigskip

\noindent  So $$E[\displaystyle {G}^2[S]]\leq \frac{p-1}{2}\sum_{v\in V(G)}\gamma(v)+\frac{\Delta}{2}|S|\leq \frac{(p-1)\Delta^2}{2p}|S|+\frac{\Delta}{2}|S|.$$ So we have that the sum of the degrees in 
 $ {G}^2[S]$ is at most $$\displaystyle \frac{(p-1)\Delta^2}{p}|S|+\Delta|S|$$ which implies that there is a vertex with degree at most $$\displaystyle \lfloor\frac{p-1}{p}\Delta^2\rfloor+\Delta $$ in $ {G}^2[S]$.
 \end{proof}
 
\bigskip 

The above result is a generalization of a result due to Gon{\c{c}}alves, Raspaud and Shalu~\cite{shalu}(see page 43, Lemma 4). 
 Now we are ready to prove the main theorem of this chapter.
 
  \bigskip

 \noindent \textit{Proof of Theorem~\ref{th rel delta}.}
 Let $\displaystyle {G}$ be an $(m,n)$-colored mixed graph with maximum degree $ \Delta$. Then $\displaystyle {G}^2$ is $\displaystyle (\lfloor\frac{p-1}{p}\Delta^2\rfloor+\Delta) $-degenerated following Lemma~\ref{th rel polta}. Let $R$ be any mixed relative clique of $ {G} $. So, $\displaystyle {G}^2[R] $, the induced graph of $ {G}^2$ by $R$, is a clique. Hence, $ {G}^2[R] $ has a vertex of degree at most $\displaystyle (\lfloor\frac{p-1}{p}\Delta^2\rfloor+\Delta) $ as $\displaystyle {G}^2 $ is $ (\lfloor\frac{p-1}{p}\Delta^2\rfloor+\Delta) $-degenerated.
 So $$|R|=| {G}^2[R]|\leq \lfloor\frac{p-1}{p}\Delta^2\rfloor+\Delta+1.$$ 
 Since we have chosen $R$ to be any mixed relative clique of $ {G}$, hence $$ \omega_{r(m,n)}({\mathcal{G}_\Delta})\leq \lfloor\frac{p-1}{p}\Delta^2\rfloor+\Delta+1. $$ 
 This completes the proof. \hfill $ \square $
 
 \medskip

 This bound is significant as the $(m,n)$-colored mixed chromatic number of $\mathcal{G}_\Delta(m,n)$ 
is known to be exponential in $\Delta$~\cite{nandi}.

\section{Subcubic Graphs}\label{Chapter4}
In this chapter, we provide the exact  
$(m,n)$-relative clique number of the family 
$\mathcal{G}_3$ of  graphs with maximum degree 3. 

We need to define a few notions used in the proof 
before stating our result. For convenience and 
self-containment of this chapter we recall a few definitions previously introduced in Chapter~\ref{Chapter3}.

\medskip

In an $(m,n)$-colored mixed graph, a vertex $a$ is said to \emph{see} a vertex $b$ if they are either adjacent or connected by a special 2-path. If $a$ and $b$ are connected by a special 2-path with $w$ acting as the internal vertex, then it is said that $u~sees~v~through~w$ or equivalently $v~sees~u~through~w$.

\medskip

We define a \emph{partial order} $\prec $ for $(m,n)$-colored mixed graphs. We define $G_1\prec G_2$ if either of the following holds:
\begin{itemize}
\item $|V(G_1)|<|V(G_2)|$,

\item $|V(G_1)|=|V(G_2)|~and~|A(G_1)|+|E(G_1)|<|A(G_2)|+|E(G_2)|$.
\end{itemize}

\medskip

Given an $(m,n)$-colored mixed graph $G$, we shall denote by $R$, a relative clique of cardinality $\omega_{r(m,n)}(G)$.
Furthermore, we denote by $S$, the set $V(G)\setminus R$. 
For convenience we call the vertices of
 $R$  $good~vertices$ and the vertices of $S$  $helper~vertices$, respectively.
 
   \vspace*{\baselineskip}
  \begin{center}
  \begin{figure}
  \begin{center}
  \includegraphics[width=\linewidth]{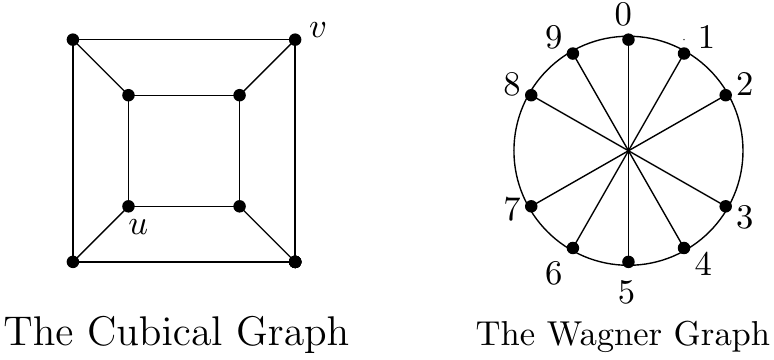}
  \caption{The cubical and Wagner Graphs}\label{fig cub-wagner}
  \end{center}
  \end{figure}
  \end{center}

  \begin{theorem}
  For the family $\mathcal{G}_3$ of $(m,n)$-colored mixed graphs with maximum degree 3, we have 
  
  \begin{itemize}
  \item[$(a)$] $ \omega_{r(1,0)}(\mathcal{G}_3)=7$~\cite{sir},
  
  \item[$(b)$] $ \omega_{r(0,2)}(\mathcal{G}_3)=7$,
  
  \item[$(c)$] $ \omega_{r(0,3)}(\mathcal{G}_3)=8$,
  
  \item[$(d)$] $ \omega_{r(0,n)}(\mathcal{G}_3)=10$ whenever $n \geq 4$,
  
  \item[$(e)$] $ \omega_{r(m,0)}(\mathcal{G}_3)=10$ whenever $m \geq 2$,
  
  \item[$(f)$] $ \omega_{r(m,n)}(\mathcal{G}_3)=10$ whenever $m,n\geq 1$.
  \end{itemize}
  \end{theorem}
  \begin{proof}
  \begin{itemize}
  
  \item[$(a)$] 	 Das,  Mj and  Sen proved this 
  part~\cite{sir}.
  
  \medskip
  
  \item[$(b)$]	  Let $ {G}$ be a minimal (with respect to $ \prec$) subcubic $(0,2)$-colored mixed graph with $\displaystyle \omega_{r(0,2)}({G})>7$ and corresponding relative clique $\displaystyle R$ where $\displaystyle R$ denotes the set of good vertices. Let $S=V(G)\setminus R$ denote the set of helper vertices. Assume that $\displaystyle |R|=r,~|S|=s,~e$ is the number of edges with both endpoints in $\displaystyle R$ and $\displaystyle t$ is the number of special 2-paths with both endpoints in $\displaystyle R$.
  
  First we shall prove that $ e\leq \frac{3}{2}(r-s)$. 
  Observe that $\displaystyle S$ is an independent set and there is no vertex of degree 1 in $\displaystyle S$ 
   as $ {G}$ is minimal. Also there is no vertex of degree 2 in $\displaystyle S$, as otherwise, we can remove that vertex and make its neighbors adjacent to obtain a graph that contradicts the fact that $ {G}$ is minimal. So, each vertex of $\displaystyle S$ is adjacent to exactly 3 good vertices. So there are $\displaystyle 3s $ edges each with exactly one endpoint in $\displaystyle S$. As $ {G}$ is subcubic
    $\displaystyle |E({G})|\leq \frac{3}{2}(r+s)$ and hence $\displaystyle e\leq |E({G})|-3s=\frac{3}{2}(r-s)$.
  
  Now we shall show that $t\leq 2r-s$. Since $ {G}$ is a subcubic graph,  each vertex of $ {G}$ can be an internal vertex of at most 2 special 2-paths. So the total number of special 2-paths in $ {G}$ is at most $ 2(r+s)$. Now 
   each vertex $ v\in S$ is adjacent to exactly three good vertices, say, $ v_1,v_2~and~v_3$. Now for each $ i\in \lbrace1,2,3\rbrace$, we have counted 2 special 2-paths through $v_i$ to get the upper bound $2(r+s)$ of the total number of special 2-paths in ${G}$. Note that one of those special 2-paths must have $v$ as an endpoint. So for each edge incident to a vertex $v\in S$, there is a distinct special 2-path in ${G}$ with one of its endpoints in $S$. Hence $ t\leq 2(r+s)-3s=2r-s$.
  
  Now as $\displaystyle R$ is a relative clique, the induced graph ${G}^2[R]$ is a clique. Now each edge of ${G}^2[R]$ either corresponds to an edge with both endpoints in $R$ or to a special 2-path with both endpoints in $R$. Hence
    $$ \frac{r(r-1)}{2}\leq e+t\leq \frac{7r-s}{2} \Rightarrow r\leq 8-\frac{s}{r}~~\cdots (\ast)$$
  
  If either $s>0$ or $e< \frac{3}{2}(r-s)$, then by relation $(\ast)$, $r\leq7$. If $s=0,r=8~ and~e=\frac{3r}{2}$. So ${G}$ is a cubic $(0,2)$-absolute clique.
  Also if $ \frac{r(r-1)}{2} < e+t$, then $r\leq 7$. So $ \frac{r(r-1)}{2}= e+t$, that is,  the endpoints of the special 2-paths are non-adjacent. So the graph ${G}$ is triangle-free.

  It is known that there are exactly two non-isomorphic triangle-free cubic graphs on 8 vertices namely the cubical graph and the Wagner graph (see Fig.~\ref{fig cub-wagner}).

  In the cubical graph, the vertices $u$ and $v$ are neither adjacent nor connected by a 2-path, so it can not 
  be the underlying undirected graph of a $(0,2)$-absolute clique. 
  
  In the Wagner Graph suppose that 04 be an edge of color 1. Then for 0 to see 3 and 5, the edges 43 and 45 must be of color 2. Now 3 and 5 are neither adjacent nor connected by a 2-path. So ${G}$ is not an underlying undirected graph of a $(0,2)$-absolute clique, a contradiction.
  Thus $r\leq 7$.

  \begin{figure}
  $$\includegraphics[scale=2]{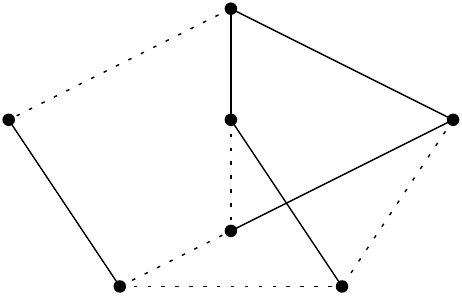}$$
  \caption{Example of a relative clique for $(m,n)=(0,2)$}\label{fig 02clique}
  \end{figure}

  Figure.~\ref{fig 02clique} shows that there exists a
 $(0,2)$-absolute clique  
of order 7. Hence $\omega_{r(0,2)}(\mathcal{G}_3)=7$.
  
   \medskip
  
  \item[$(c)$]	  Let ${G}$ be a minimal (with respect to $ \prec$) subcubic $(0,3)$-colored mixed graph with $\displaystyle \omega_{r(0,3)}(\mathcal{G}_3) > 8$ and corresponding relative clique $\displaystyle R$, where $\displaystyle R$ denotes the set of good vertices. Let $S=V(G)\setminus R$ denote the set of helper vertices. Assume that $\displaystyle |R|=r,~|S|=s,~e$ is the number of edges with both endpoints in $\displaystyle R$ and $\displaystyle t$ is the number of special 2-paths with both endpoints in $\displaystyle R$.
  
  We now prove that $ e\leq \frac{3}{2}(r-s)$. Clearly, $\displaystyle S$ is an independent set and there is no vertex of degree 1 in $\displaystyle S$  as $ {G}$ is minimal. Also there is no vertex of degree 2 in $\displaystyle S$ as otherwise, we can remove that vertex and make its neighbours adjacent to obtain a graph that contradicts the fact that $ {G}$ is minimal. So, each vertex of $\displaystyle S$ is adjacent to exactly 3 good vertices. So there are $\displaystyle 3s $ edges each with exactly one endpoint in $\displaystyle S$. As $ {G}$ is subcubic,  $\displaystyle |E({G})|\leq \frac{3}{2}(r+s)$ and hence $\displaystyle e\leq |E({G})|-3s=\frac{3}{2}(r-s)$.
  
  Now we shall show that $t\leq 3(r-s)$. Since $ {G}$ is a subcubic graph, so each vertex of ${G}$ can be an internal vertex of at most 3 special 2-paths. So the total number of special 2-paths in $ {G}$ is at most $ 3(r+s)$. Now each each vertex $ v\in S$ is adjacent to exactly 3 good vertices, say, $ v_1,v_2~\text{ and }~v_3$. Now for each $ i\in \lbrace1,2,3\rbrace$, we have counted 3 special 2-paths through $v_i$ to get the upper bound $3(r+s)$ of the total number of special 2-paths in ${G}$. Clearly one of those special 2-paths must have $v$ as an endpoint. So for each edge incident to a vertex 
  $v\in S$, there is a distinct special 2-path in ${G}$ with one of its endpoints in $S$. Hence $ t\leq 3(r+s)-6s=3(r-s)$.
  
  Now as $\displaystyle R$ is a relative clique, the induced graph ${G}^2[R]$ is a clique. Now each edge of ${G}^2[R]$ either corresponds to an edge with both endpoints in $R$ or to a special 2-path with both endpoints in $R$. Hence, $$ \frac{r(r-1)}{2}\leq e+t\leq \frac{3}{2}(r-s)+3(r-s) \Rightarrow r\leq 10-9\frac{s}{r}~~...(\ast)$$
  
 If either $s>0$ or $e< \frac{3}{2}(r-s)$, then by relation $(\ast)$, $r\leq 9$. If $s=0,r=10~ \text{ and }~e=\frac{3r}{2}$. So ${G}$ is a cubic $(0,3)$-absolute clique.  
 Also if $ \frac{r(r-1)}{2}< e+t$, then $r\leq 9$. 
 So $ \frac{r(r-1)}{2}= e+t$ if $r=10$, that is,  the endpoints of the special 2-paths are non-adjacent. So the graph ${G}$ is triangle-free.
 
 It is known that there is exactly one non-isomorphic triangle-free cubic graph with diameter 2 on 10 vertices namely the Petersen graph. In the Petersen graph there exists a unique 2-path between any two vertices. Thus the Petersen graph 
 is underlying graph of a $(0,3)$-absolute clique if and only if the Petersen graph is 3-edge-colorable. However, we know that the Petersen graph is not 3-edge-colorable~\cite{reza}. 
 Thus $\displaystyle \omega_{r(0,3)}(\mathcal{G}_3) \leq 9$.
 
So $|R|=r\leq 9$. Putting $r=9$ in relation $\ast$, we have $s\leq 1$. So there exists at most 1 helper vertex.

So assume that there exists $G$ such that $\displaystyle \omega_{r(0,3)}(G)=9$ and let $ R$ be the relative clique of order 9. 

\begin{figure}
 $$\includegraphics[scale=0.5]{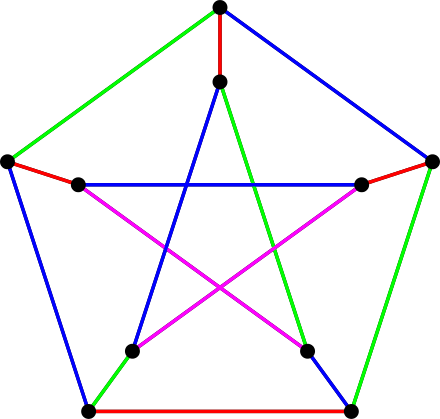}$$
 \caption{Example of a relative clique for $(m,n)=(0,4)$}\label{fig 04clique}
 \end{figure}

However, a graph on 9 vertices can not be $3$-regular due to the Handshaking lemma. Thus $H$ must have exactly 10 vertices.

A good vertex of $G$ must see the other 8 good vertices directly or through a 2-path.
Observe that  there can not be a good vertex of degree 1 or 2, as otherwise, that vertex will not be able to see all the other good vertices of $G$. So all the good vertices 
 must be of degree 3. Moreover, we have already noticed that a helper vertex must have degree 3 due to the minimality of $G$. Thus, $G$ is a diameter 2 cubic graph on 10 vertices
 with exactly 9 good vertices and 1 helper. We know that the Petersen graph is the only graph which satisfies this property. 
 
 As we can always 
 color the incident edges of a helper vertex with 3 different colors without decreasing the $(0,3)$-relative clique number of $G$, we can again conclude that 
  the Petersen graph 
 is underlying graph of a $(0,3)$-colored mixed graph with $(0,3)$-relative clique of cardinality 9 if and only if the Petersen graph is 3-edge-colorable. However, we know that the Petersen graph is not 3-edge-colorable~\cite{reza}. 
Thus $\displaystyle \omega_{r(0,3)}(G)\leq8$.

 We now show that 
 $\displaystyle \omega_{r(0,3)}(\mathcal{G}_3) \geq 8$. 
 Consider the Wagner Graph. We mark the chords of the circle as type 1 edge and the eight subdivisions of the circumference as types 2 and 3 respectively. The Wagner Graph now becomes a relative clique. 
 
 So
  $\displaystyle \omega_{r(0,3)}(\mathcal{G}_3))=8$.
 
\medskip

 \item[$(d)$]	Proceeding as in part (c), it can be shown that $ \omega_{r(0,n)}(\mathcal{G}_3)\leq 10$ for all $n \geq 4$. 
  Fig.~\ref{fig 04clique} shows that the Petersen Graph is a $(0,n)$-absolute clique for all $n \geq 4$. 
   So $ \omega_{r(0,n)}(\mathcal{G}_3)=10$ whenever $n\geq 4$.

\begin{figure}
 $$\includegraphics[scale=.4]{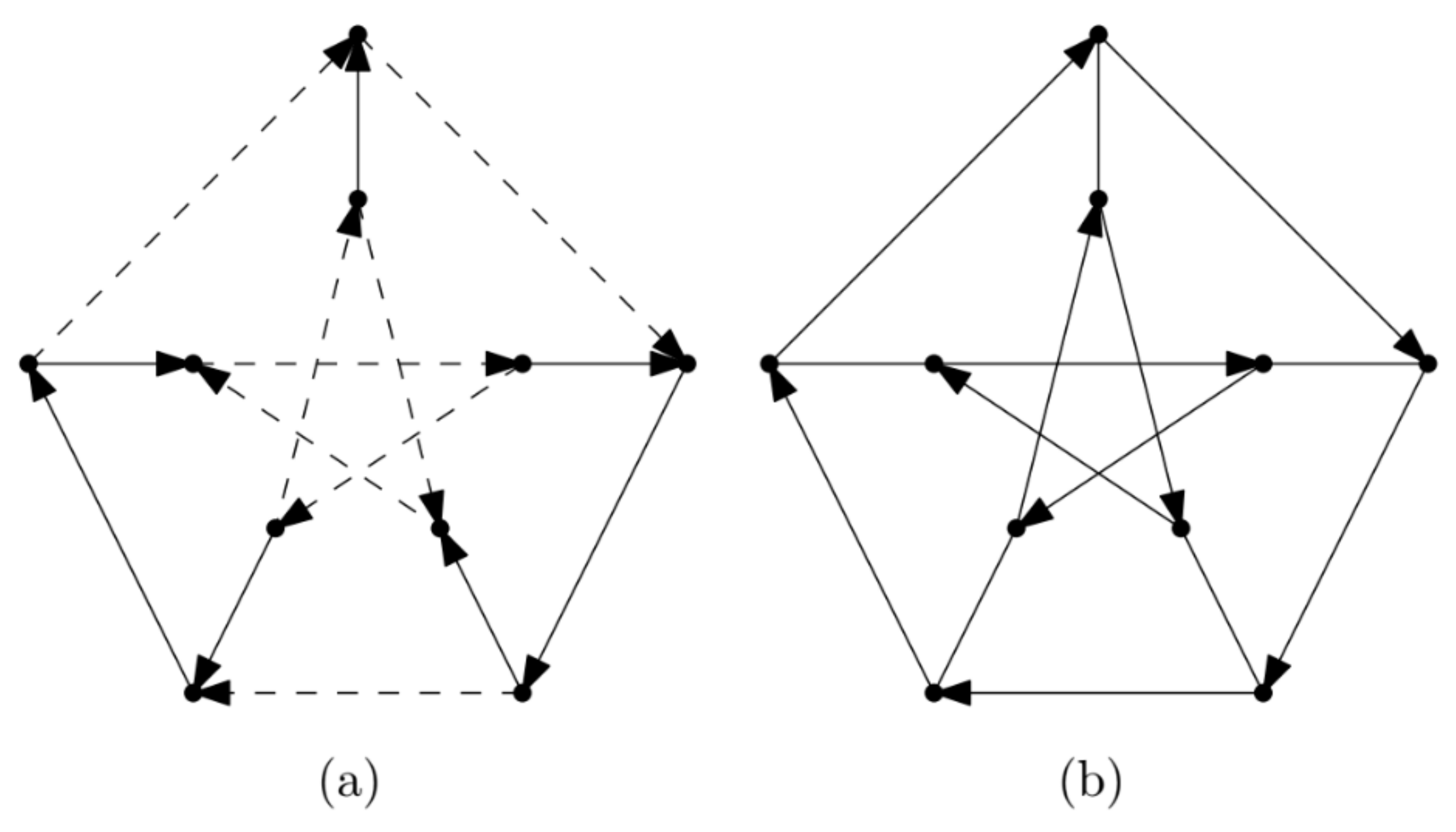}$$
 \caption{(a) Example of a relative clique for $(m,n)=(2,0)$ (b) Example of a relative clique for $(m,n)=(1,1)$}\label{fig 20clique}
 \end{figure}

\medskip 
 
 \item[$(e)$] Proceeding as in part (c), it can be shown that $ \omega_{r(m,0)}(\mathcal{G}_3)\leq 10$ for all $m \geq 2$. 
 Fig.~\ref{fig 20clique} shows that $ \omega_{r(2,0)}(\mathcal{G}_3) \geq 10$.
 Hence $ \omega_{r(m,0)}(\mathcal{G}_3)=10$ whenever $m\geq 2$.

\medskip 
 
 \item[$(f)$]	Proceeding as in part (e), we see that $ \omega_{r(m,n)}(\mathcal{G}_3)\leq 10$ for all 
 $m,n \geq 1$. 
 Fig.~\ref{fig 20clique} shows that $ \omega_{r(1,1)}(\mathcal{G}_3)\geq 10$.
Hence, $ \omega_{r(m,n)}(\mathcal{G}_3)=10$ whenever $m,n\geq 1$. 
 \end{itemize}
\end{proof}

The above result is a generalization of a result due to Das, Mj ans Sen~\cite{sir} where they studied only the case $(m,n)=(1,0)$.

\section{Planar Graphs}\label{Chapter5}
A \textit{planar graph} is a graph that can be embedded in the plane, that is, it can be drawn on the plane in such a way that its edges intersect only at their endpoints. In other words, it can be drawn in such a way that no edges cross each other.
  
  We give an upper bound for the $(m,n)$-relative clique number of the family $ \mathcal{P}$ of  planar graphs.
  
We need to define a few notions used in the proof 
before stating our result. For convenience and 
self-containment of this chapter we recall a few definitions previously introduced in Chapter~\ref{Chapter3} 
and~\ref{Chapter4}.

\medskip

The set of all vertices adjacent to a vertex $v$  is denoted by $N(v)$. The set of vertices adjacent to a vertex $v$ by an edge of color $i$ in an $(m,n)$-colored mixed graph $G$ 
is denoted by $N_i(v)$ for all $i\in \lbrace1,2,3,\cdots,n\rbrace $; the set of vertices $u$ adjacent to $v$ with an arc $vu$ of color $j$ is denoted by $N_j^{+}(v)$ and the set of vertices $u$ adjacent to $v$ with an arc $uv$ of color $j$ is denoted by $N_j^{-}(v)$ for all $j\in \lbrace 1,2,3,\cdots,m \rbrace $.
So $\displaystyle N(v)=\lbrace \bigcup^{n}_{i=1} N_i(v)\rbrace \bigcup \lbrace \bigcup^{m}_{j=1} N_j^{+}(v) \rbrace \bigcup \lbrace \bigcup^{m}_{j=1} N_j^{-}(v) \rbrace $. The $closed~neighborhood$ of $v$ is defined as 
$N[v]=N(v)\cup \{v\}$. Let $d(v)=|N(v)|$ denote the $degree$ of $v$.

 \medskip
 
Two vertices $a$ and $b$ of an $(m,n)$-colored mixed graph are said to $agree$ on another vertex $c$ either if $ac$ and $bc$ are edges of the same color or if $c\in N_i^\gamma (a) \cap N_i^\gamma (b)$ for some $\gamma \in \lbrace+,-\rbrace$ for some $i\in \lbrace 1,2,3,\cdots,m\rbrace$. The vertices $a$ and $b$ are said to $disagree$ on $c$ if they don't $agree$ on $c$. A vertex $a$ is said to see a vertex $b$ if they are either adjacent or connected by a special 2-path. If $a$ and $b$ are connected by a special 2-path with $w$ acting as the internal vertex, then it is said that $u~sees~v~through~w$ or equivalently $v~sees~u~through~w$.

\medskip

We now define a \emph{partial order} $\prec $ for $(m,n)$-colored mixed graphs. We define $G_1\prec G_2$ if either of the following holds:
\begin{itemize}
\item $|V(G_1)|<|V(G_2)|$,

\item $|V(G_1)|=|V(G_2)|~and~|A(G_1)|+|E(G_1)|<|A(G_2)|+|E(G_2)|$.
\end{itemize}

Now we are ready to state and prove the main result of this chapter.

 \begin{theorem}
 For the family $ \mathcal{P}$ of  planar graphs, 
 we have $$ 3p^2+p+1\leq\omega_{r(m,n)}(\mathcal{P})\leq42p^2+8$$  where $ p=2m+n >1$.
 \end{theorem}
 
 \begin{proof}
 The lower bound follows from the fact that 
 $\omega_{r(m,n)}(\mathcal{P}) \geq \omega_{a(m,n)}(\mathcal{P})\geq 3p^2+p+1$ \cite{duffy}.
 
 Let $ {G}$ be a minimal (with respect to $\prec $) $(m,n)$-mixed planar graph with $ \omega_{r(m,n)}({G})\geq42p^2+9$ with a maximum $(m,n)$-relative clique $R$, where $R$ denotes the set of good vertices.
 Also, let $ S=V({G})\setminus R$ denote the set of helper vertices.
 
 Let $h$ be a helper vertex with degree at most 3. So we add an arc or an edge between each non-adjacent pair of vertices from $N(h)$ and delete $h$, resulting in a graph that is still planar and has $(m,n)$-relative clique number 
 $ \omega_{r(m,n)}({G})$ contradicting that $ {G}$ is minimal. So a helper vertex $h\in S$ has degree at least 4.
 
 We now show some forbidden configurations for mixed planar graphs with respect to vertices of its relative clique.
 
 $ \langle1\rangle$ Forbidden configuration $F_1:4m+2n+1=2p+1 $ independent vertices of a relative clique agreeing on a vertex of an $(m,n)$-mixed planar graph through an edge or an arc.
 \begin{proof}
 Assume that $ 4m+2n+1=2p+1$ independent vertices $ a_1,a_2,a_3,\cdots,a_{4m+2n+1}$ of a relative clique agree on a vertex $a$ of an $(m,n)$-mixed planar graph through an edge or an arc in clockwise order around $a$ in a fixed embedding of $G$. Since $ a_i~ and ~a_{i+2}$ are independent vertices of a relative clique, they must see each other through a vertex $w$. Observe that this vertex $w$ is different from the vertex $a$. So $a_i$ sees $a_{i+2}$ through $w$ with the help of a special 2-path. Now $a_{i+1}$ must see the vertices of the relative clique other than $ a_i~ and ~a_{i+2}$ through $w$ to preserve the planarity of the graph. So $w$ is adjacent to all the $4m+2n+1$ vertices. Since there are $m$ arcs and $n$ edges, by Pigeon Hole Principle, there must be three vertices that agree on $w$ through an edge/arc.

  Let $ a_{p},a_{q}~and~a_{r} $ be three such vertices, agreeing on $w$ through an edge/arc. Now a given vertex among the $4m+2n+1$ vertices must see at least $4m+2n-2$ of the remaining vertices through $w$ to preserve the planarity of the graph. So if $a_p$ sees the $4m+2n-2$ vertices other than $a_q$ and $a_r$ through $w$, then the edges/arcs $ aa_{p},aa_{q}~and~aa_{r} $ must be consecutive with the edge/arc $aa_p$ being in the middle. Then $a_q$ cannot see $a_r$ without destroying the planarity of $G$.
\end{proof}

\bigskip

 $ \langle2\rangle$ Forbidden configuration $F_2:2p^2+1$ independent vertices of a relative clique adjacent to a vertex of an $(m,n)$-mixed planar graph.
 \begin{proof}
 At least $2p+1$ of the $2p^2+1$ independent vertices must agree on their common neighbour creating $F_1$.
 \end{proof}
 
\bigskip 
 
 Jendrol' and Voss \cite{voss} showed that any planar graph with minimum degree at least 4 must have an edge/arc $xy$ such that $d(x)+d(y)\leq 11$. Therefore there exists $v\in R$ with $N(v)=\{v_1,v_2,\cdots,v_r\}$ where $r\leq 7$. We fix such a vertex $v$ for the rest of this proof. Now each vertex of $R\diagdown N[v]$ sees $v$ through a vertex from $N[v]$.
 
 So the vertices of $R\diagdown N[v]$ induce an outer-planar graph and is hence 3-colorable. So $R\diagdown N[v]$ contains a set $I$ of at least $14p^2+1$ independent vertices.
 We disjointly partition the vertices of $\displaystyle I = \bigcup^{r}_{i = 1} S_i$ such that the vertices of $S_i$ sees $v$ through $v_i$ but not through $v_j$ for any $j<i$ where $\displaystyle i,j\in \{1,2,3,\cdots,r\}$. Assume that the vertices of $N(v)$ are indexed in such a way that the quantity ($|S_1|,|S_2|,|S_3|,\cdots,|S_r|$) is maximum with respect to lexicographic ordering. Now $|S_1| \geq 2p^2+1$ by pigeon hole principle and $|S_1| \leq 2p^2$ due to forbidden configuration $F_2$ leading to a contradiction.
 
 Hence, $\omega_{r(m,n)}(\mathcal{P})\leq42p^2+8$ where $p=2m+n$.
 \end{proof}

An improvement of the above result, when restricted to the case $(m,n)= (1,0)$ has been proved by 
Das, Mj, and Sen~\cite{sir}. However, for the other values of
 $(m,n) \neq (1,0)$ 
our result is the best known. 

\section{Conclusions}\label{chap con}
This article can be considered as the first systematic study of $(m,n)$-relative clique number of colored mixed graphs. Here we studied the parameter for the graphs with bounded maximum degree, subcubic graphs and planar graphs. One can further extend our work by improving the bounds or by studying the parameter for other families of graphs.

\bibliographystyle{abbrv}
\bibliography{example}

\end{document}